\documentclass[english,aps,manuscript]{revtex4}
\usepackage[OT1]{fontenc}
\usepackage[latin9]{inputenc}
\usepackage{amsthm}
\usepackage{amsmath}
\usepackage{graphicx}
\usepackage{amssymb}

\makeatletter
\@ifundefined{textcolor}{}
{%
 \definecolor{BLACK}{gray}{0}
 \definecolor{WHITE}{gray}{1}
 \definecolor{RED}{rgb}{1,0,0}
 \definecolor{GREEN}{rgb}{0,1,0}
 \definecolor{BLUE}{rgb}{0,0,1}
 \definecolor{CYAN}{cmyk}{1,0,0,0}
 \definecolor{MAGENTA}{cmyk}{0,1,0,0}
 \definecolor{YELLOW}{cmyk}{0,0,1,0}
 }
\numberwithin{equation}{section}
\numberwithin{figure}{section}
\theoremstyle{plain}
\newtheorem{thm}{Theorem}
  \theoremstyle{definition}
  \newtheorem{defn}[thm]{Definition}
 \theoremstyle{definition}
  \newtheorem{example}[thm]{Example}
  \theoremstyle{plain}
  \newtheorem{prop}[thm]{Proposition}
  \theoremstyle{plain}
  \newtheorem{cor}[thm]{Corollary}

\makeatother

\usepackage{babel}

\begin{document}

\title{Extending and Characterizing Quantum Magic Games}

\author{Alex Arkhipov}

\email{arkhipov@mit.edu}

\affiliation{MIT}
\begin{abstract}
The Mermin-Peres magic square game \cite{Mermin,Peres,Aravind description}
is a cooperative two-player nonlocal game in which shared quantum
entanglement allows the players to win with certainty, while players
limited to classical operations cannot do so, a phenomenon dubbed
{}``quantum pseudo-telepathy''. The game has a referee separately
ask each player to color a subset of a $3\times3$ grid. The referee
checks that their colorings satisfy certain parity constraints that
can't all be simultaneously realized.

We define a generalization of these games to be played on an arbitrary
arrangement of intersecting sets of elements. We characterize exactly
which of these games exhibit quantum pseudo-telepathy, and give quantum
winning strategies for those that do. In doing so, we show that it
suffices for the players to share three Bell pairs of entanglement
even for games on arbitrarily large arrangements. Moreover, it suffices
for Alice and Bob to use measurements from the three-qubit Pauli group.
The proof technique uses a novel connection of Mermin-style games
to graph planarity.
\end{abstract}
\maketitle

\section{Introduction}

Our goal is to generalize the Mermin-Peres magic square and magic
pentagram game \cite{Mermin,Peres,Mermin pentagram} to be played
on more general configurations of points and lines that we will call
\textit{arrangements}. We characterize exactly which of these games
allow two quantum players to win with certainty while two classical
players cannot do so.

We will begin by introducing the Mermin magic square game and related
prior work. Afterwards, we define more general terms for talking about
magic games that will let us extend the proofs for the magic square
games to a more general class of games. Our main result is to characterize
which of these magic games are winnable by a quantum strategy, construct
the winning strategy for those games and a proof of impossibility
for the others. The construction will imply a bound on the resources
needed to win such a game.

\subsection{Quantum Nonlocal Games}

The Mermin-Peres magic-square game and the generalizations we will
study are examples of \textit{nonlocal games}. See \cite{Nonlocal}
for an overview including formal definitions. 

Nonlocal games are played by some number of cooperating players and
moderated by a referee. The game protocol is as follows: the referee
randomly asks each player a question selected from some joint distribution,
the players respond to the referee independently without communicating,
and the referee decides whether the players win or lose based on the
players' answers. For our result, we will look only at two-player
nonlocal games.

A \textit{quantum strategy} allows the players to share entanglement
prior to the game and to perform local quantum operations and measurements,
in contrast to a \textit{classical strategy}, which does not. Note
that in both cases, the game played is the same; the difference is
only in what operations the players are allowed to use to produce
their answers. Sharing randomness prior to the game cannot help classical
players, as shown by a standard convexity argument.

For some nonlocal games, players using a quantum strategy can win
with a higher probability than possible with any classical strategy.
Examples include the CHSH game and the Kochen-Specker Game \cite{Nonlocal}.
It is even possible for a quantum strategy to win with certainty,
in which case we call it a \textit{winning strategy}, even when no
classical winning strategy exists, an effect called quantum pseudo-telepathy
\cite{Nonlocal}.

Such results allow empirical demonstrations of the power of quantum
operations to correlate outcomes of causally separated events in a
classically impossible way, thereby ruling out local hidden variables.
Quantum nonlocal games demonstate the power of quantum operations
in a context without the quantitative resource bounds that commonly
appear in quantum information and quantum computing.

\subsection{Mermin-Peres Magic Square Game}

The Mermin-Peres magic square game \cite{Mermin,Peres,Aravind description}
is a two-player nonlocal game in which shared quantum entanglement
allows the players to win with certainty, while players limited to
classical operations cannot do so, a phenomenon dubbed {}``quantum
pseudo-telepathy''.
\begin{defn}
The \textbf{Mermin-Peres magic square game} is a two-player cooperative
game played by the protocol described below. The players Alice and
Bob may agree on a prior strategy in advance, but cannot communicate
once the game starts. 
\begin{enumerate}
\item The referee picks a random row $r\in\left\{ 1,2,3\right\} $ and a
random column $c\in\left\{ 1,2,3\right\} $
\item The referee sends $r$ to Alice and $c$ to Bob. 
\item Alice colors each of three cells in row $r$ in a $3\times3$ grid
either red or green, and sends this coloring to the referee.
\item Bob colors each of three cells in column $c$ in a $3\times3$ grideither
red or green, and sends this coloring to the referee.
\item The referee checks that Alice has colored red an even number of cells.
\item The referee checks that Bob has colored red an odd number of cells.
\item The referee checks that Alice and Bob have assigned the same color
to the cell in row $r$ and column $c$.
\item If all these checks succeed, then Alice and Bob win the game; otherwise
they lose.
\end{enumerate}
\end{defn}
Even though this game is defined purely as a classical game, players
can gain an advantage by using prior quantum entaglement and performing
quantum operations.
\begin{thm}
There is a quantum strategy by which Alice and Bob win the magic square
game with certainty, but no such classical strategy.
\end{thm}
We will not prove this here and instead wait to prove the result for
a more general class of games. A proof can be found in \cite{Aravind description}.
The key to winning the game comes from a construction Mermin gave
to prove a version of the Kochen-Specker Theorem\cite{Mermin}, which
has the following properies.
\begin{thm}
\label{thm:Construction magic square}There is a labeling of each
cell of the $3\times3$ squares with a quantum observables such that
\begin{itemize}
\item The eigenvalues of each observable are all $+1$ or $-1$.
\item The observables in each row commute and multiply to $+I$.
\item The observables in each column commute and multiply to $-I$.
\end{itemize}
Moreover, there is no way to do this {}``classically'' using observables
of the form $\pm I$.
\end{thm}

\subsection{Generalizations of the Magic Square Construction}

Mermin constructed another example of a quantum telepathy game called
the magic pentagram game \cite{Mermin pentagram}. This game is played
on an arrangement of ten points joined by five lines of four points
each, arranged like a five-sided star. The result \cite{Triangle}
looked at a different arrangement, a subset of the Fano plane, and
proved that there is no such quantum winning strategy for this game.
These examples suggest the generalized notion of magic games that
is explored in this paper, and raise the question of exactly which
of these games are quantum-winnable.

Prior research into generalizing Mermin's constructions focused on
understanding the observables used to win the magic square and magic
pentagram games, and finding all possible sets of such observables.
The results of \cite{Square geometry,Pentagram geometry} interpret
Mermin's construction in terms of geometrical structures on finite
rings. Our work, however, seeks to determine when there exist winning
strategies for generalized arrangement rather than to classify all
winning strategies for existing arrangements.

\subsection{Binary constraint systems}

Recently, Cleve defined games on binary constraint systems\cite{BCS},
which are a yet more general notion than the Mermin-style games defined
in this paper. A binary constraint system consists of a finite set
of constraints over finitely many binary variables. The corresponding
game has the referee pick a variable and a constraint containing it,
and ask one player to assign a value to that variable, and the other
to do so for all variables appearing in the constraint. In this view,
our generalization of Mermin games are binary constraint systems with
parity constraints, and with the restriction that every variable appears
exactly twice.

The assignment of measurements that defines the winning strategy for
the Mermin square game can be thought to be a quantum solution to
the constraints. The result \cite{BCS} shows that the existence of
a quantum solution to the constraints is both necesarry and sufficient
for the constraint game to be quantum-winnable.

Therefore, finding quantum strategies to win Mermin-style games completely
reduces to finding operators that satisfy constraints for the game
as was done in Theorem \ref{thm:Construction magic square} for the
magic square. In other words, to win a Mermin game with certainty
using shared entanglement, it suffices to construct a strategy analogous
to Mermin's strategy for the magic square game. 

Cleve's result \cite{BCS} strengthens our main result. Our results
characterize for which generalized Mermin game there exists a satisfying
operator assignment, which implies a quantum winning strategy, but
does not rule out other alternative strategies to win such games.
Combined with the result \cite{BCS} equating quantum winnability
with existence of a satisfying operator assignment, our result gives
an exact characterization of which generalized Mermin games have a
quantum winning strategy but no classical one. 

The paper \cite{BCS} also gives an ad-hoc substitution approach from
Speelman\cite{Triangle} that can be used to prove the inconsistency
of a set of parity constraints. One of our results (Theorem \ref{thm:planar implies nonmagic})
uses a similar technique in its proof, cancelling variables by contracting
edges in a planar graph corresponding to the constraints. Such a graph-theoretic
interpretation of the substitution approach resolves the question
of when it can be applied to show an inconsistency.

\section{Preliminaries and Definitions}

\subsection{Arrangements and Realizations}

The magic square and pentagram are examples of configurations on which
Mermin-style games may be played. We will call these \textit{arrangements}.
\begin{defn}
An \textbf{arrangement }$A=\left(V,\thinspace E\right)$ is a finite
connected hypergraph with vertex set $V$ and hyperedge set $E$,
where a hyperedge is a nonempty subset of $V$, such that each vertex
lies in exactly two hyperedges (\textit{connected} means the hypergraph
can't be split into two smaller disjoint hypergraphs). A \textbf{signed
arrangement} $A=\left(V,\thinspace E,\thinspace l\right)$ also contains
a labelling $l:\, E\rightarrow\left\{ +1,-1\right\} $ of each hyperedge
in $E$ with a sign of $+1$ or $-1$.
\end{defn}
We'll often represent an arrangement by drawing each vertex as a \textit{point}
and each hyperedge as a line, arc, or circle that passes through the
points it contains, which we'll generally call a \textit{line}. In
a signed arrangement, each line is labelled by $+1$ or $-1$. Note
that where the points are drawn and how they are ordered on a line
is immaterial; it only matters which points share a line.
\begin{example}
Two well-studied arrangements, the magic square and magic pentagram
are shown in Figure \ref{fig:Magic square and pentagram} with signs.
Note that while in both every hyperedge contains an equal number of
vertices, this need not be the case in general.

\begin{figure}
\includegraphics[scale=0.5]{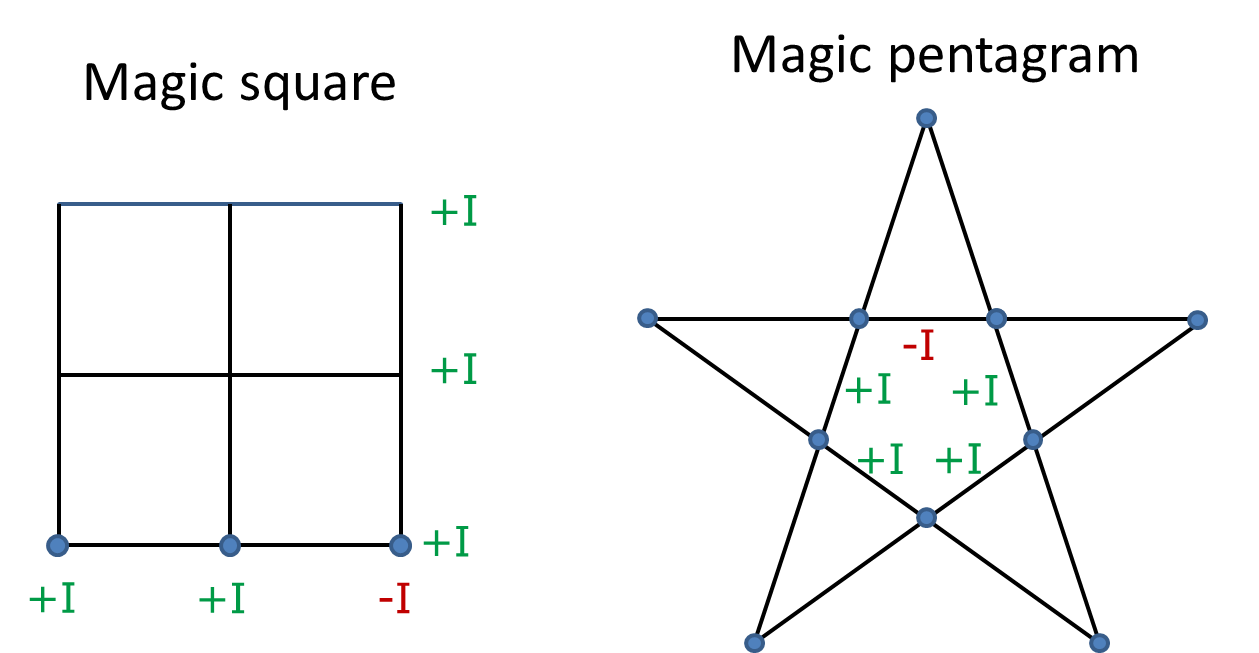}

\caption{\label{fig:Magic square and pentagram}}

\end{figure}
\end{example}
\begin{defn}
A \textbf{classical realization} of a signed arrangement $A=\left(V,\thinspace E,\thinspace l\right)$
is a labelling $c:\, V\rightarrow\left\{ +1,-1\right\} $ of the vertices
so that the product of the labels on vertices within any hyperedge
equals the label of that hyperedge:\[
\prod_{u\in e}c\left(u\right)=l\left(e\right)\mbox{ for each }e\in E\]

\end{defn}

\begin{defn}
A \textbf{quantum realization} of a labelled arrangement $A=\left(V,\thinspace E,\thinspace l\right)$
is a labelling $c:\, V\rightarrow GL\left(\mathcal{H}\right)$ of
the vertices with observables on a fixed finite-dimensional Hilbert-space
$\mathcal{H}$ such that:
\begin{itemize}
\item The observable $M$ assigned to any vertex is Hermitian and squares
to the identity ($M^{2}=I$), or equivalently, each observables orthogonally
diagonalizes with eigenvalues of $+1$ and $-1$.
\item For each hyperedge, the observables assigned to its vertices pairwise
commute.
\item For each hyperedge, the product of of the observables assigned to
its vertices equals either the identity in $\mathcal{H}$ or its negation
according to the sign of that hyperedge \[
\prod_{u\in e}c\left(u\right)=l\left(e\right)I\mbox{ for each }e\in E\]

\end{itemize}
\end{defn}
We'll say that an arrangement is \textbf{classically realizable }if
it has a classical realization and likewise \textbf{quantum realizable}
if it has a quantum realization. We note that a classical realization
is simply a quantum realization in which $\mathcal{H}=\mathbb{R}$;
therefore, every classically realizable arrangement is quantumly realizable.
An example of a quantum realization is given in Figure \ref{fig:Magic square and pentagram with operators}.

\begin{center}
\begin{figure}[h]
\includegraphics[scale=0.5]{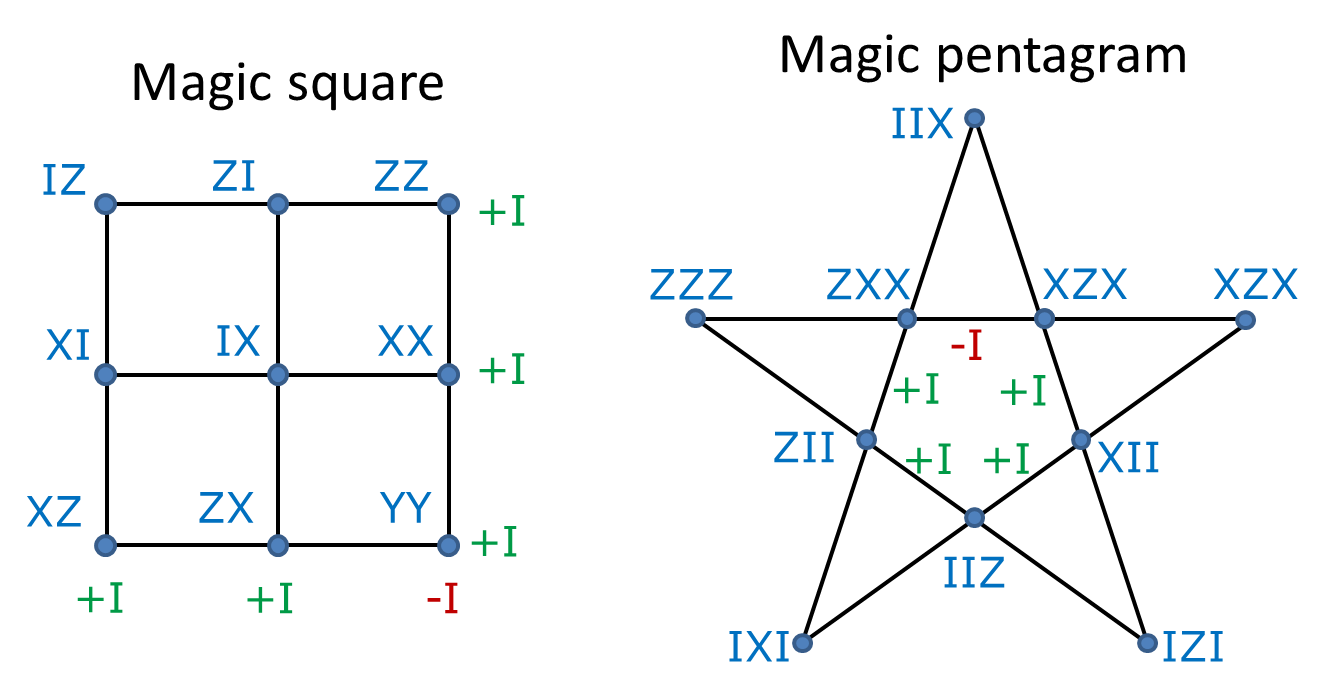}

\caption{Quantum realization of the magic square arrangement and magic pentagram
arrangements. The strings of symbols $I,X,Y,Z$ represent tensor product
of Pauli matrices. \label{fig:Magic square and pentagram with operators}}

\end{figure}

\par\end{center}

In order to take advantage of the extra freedom in constructing a
quantum realization, one must take advantage of noncommuting observables,
since using commuting observables gives no extra power beyond classical
realizability.
\begin{prop}
If a signed arrangment is quantumly realizable with observables that
all mutually commute, then it is classically realizable.\end{prop}
\begin{proof}
Commuting observables are mutually diagonalizable. Replacing each
observable by its diagonalization via conjugation gives a new quantum
realization where the observables are diagonal matrices with diagonal
entries of $\pm1$. For any basis index $i$, the $\left(i,i\right)$
entries of each diagonalized observable give a classical realization
of the arrangement.
\end{proof}

\subsection{Sign Parities}

It turns out that whether an signed arrangement is classically realizable
or quantumly realizable depend on the hyperedge signs only in a limited
manner, as one can adjust the signs of the realization operators to
achieve different signs hyperedges signs. The only salient feature
of the hyperedge signs is whether the number of $-1$ labels is odd
or even.
\begin{defn}
The \textbf{parity }$p\left(l\right)$ of the signing $l$ of an arrangement
$A=\left(V,\thinspace E,\thinspace l\right)$ is \[
p\left(l\right)=\prod_{e\in E}l\left(e\right),\]
which is $-1$ if there's an odd number of $-1$ labels, and $+1$
if there's an even number. \end{defn}
\begin{prop}
\label{pro:Classical parity}The classical realizability of a signed
arrangement $A=\left(V,\thinspace E,\thinspace l\right)$ depends
on $l$ only via its parity $p\left(l\right)$. In other words, $A'=\left(V,\thinspace E,\thinspace l'\right)$
has the same classical realizability as $A'=\left(V,\thinspace E,\thinspace l'\right)$
if $p\left(l'\right)=p\left(l\right)$.\end{prop}
\begin{proof}
Suppose $A$ is classically realizable, and let $c$ be its classical
realization. We will construct a corresponding classical realization
of $c'$. 

First, we show that the result holds when $l'$ is achieved by flipping
the signs $l$ assigns to two hyperedges, $a$ and $b$. We note that
flipping the label of a vertex flips the parity products of the two
hyperedges containing it. Since an arrangement is finite and connected,
there must be a path $e_{0},e_{1},\dots e_{n}$ of distinct edges
starting at $e_{0}=a$ and ending at $e_{n}=b$ such that any pair
of hyperedges $e_{i},e_{i+1}$ adjacent in the sequence intersects
at a vertex $v_{i}$. Then, negating the label of every vertex $v_{i}$
in the path\[
l'\left(v\right)=\begin{cases}
-l\left(v\right), & \mbox{if }v\in\left\{ v_{0},\dots,v_{n+1}\right\} \\
l\left(v\right), & \mbox{otherwise}\end{cases}\]
achieves the desired result: The product of the labels of the vertices
on edge $e_{i}$ is unaffected, as two vertices within it have flipped
labels, except the edges $e_{0}=a$ and ending at $e_{n}=b$ at the
ends of the chain.

By repeatedly changing the realization to flip pairs of signs in $l$,
one can go from any labelling to any other labelling of equal parity.\end{proof}
\begin{prop}
\label{pro:Classical realizable implies even}A classical arrangement
$A=\left(V,\thinspace E,\thinspace l\right)$ is realizable if and
only if the parity $p\left(l\right)$ is $+1$.\end{prop}
\begin{proof}
The even-parity signing where each hyperedge has sign $+1$ is realized
by assigning $+1$ to each vertex. Then, by Proposition \ref{pro:Classical parity},
any even-parity signing is classically realizable. 

No odd-parity signing is realizable, since each vertex lies in two
hyperedges, so the product of the vertex labels of each edge will
contain each vertex label twice. 

\[
\prod_{e\in E}l\left(e\right)=\prod_{e\in E}\prod_{v\in e}c\left(e\right)=\prod_{v\in V}c\left(e\right)^{2}=\prod_{v\in V}1=1\]

\end{proof}
A similar result to Proposition \ref{pro:Classical parity} follows
for quantum realizability. 
\begin{prop}
\label{pro:Quantum parity}The quantum realizability of a signed arrangement
$A=\left(V,\thinspace E,\thinspace l\right)$ depends on $l$ only
via its parity $p\left(l\right)$.\end{prop}
\begin{proof}
The proof is the same as that of Proposition \ref{pro:Classical parity},
except we also check that negating the quantum observables does not
change them having order two or mutually commuting within each hyperedge.
\end{proof}
In light of the results of Propositions \ref{pro:Classical parity}
and Proposition \ref{pro:Quantum parity}, we should think of quantum
realizability as a property of an unsigned arrangement.
\begin{defn}
An arrangement is \textbf{magic} if it has an odd-parity signing that
is quantumly realizable.
\end{defn}
So, a signed arrangement is magic if its underlying arrangement is
magic and the signing has odd parity. Note that by Proposition \ref{pro:Classical realizable implies even},
any magic signed arrangment is not classically realizable, and therefore
represents a gap in what's classically possible and what's quantumly
possible. 
\begin{thm}
\label{thm:Square and pentagram are magic}[Mermin, Peres] The magic
square and magic pentragram pentagram are magic arrangements.\end{thm}
\begin{proof}
Example quantum realizations for the magic square and pentagram are
pictured in Figure \ref{fig:Magic square and pentagram with operators},
using odd-parity signings for both of them. The magic square uses
Hilbert space $\left(\mathbb{C}^{2}\right)^{\otimes2}$ and measurement
operators from the two-qubit Pauli group, and the magic pentagram
does likewise with $\left(\mathbb{C}^{2}\right)^{\otimes3}$ and the
three-qubit Pauli group.
\end{proof}

\subsection{Parity telepathy games}

We extend the Mermin magic square game to be played on an arbitrary
arrangement.
\begin{defn}
The \textbf{parity telepathy game }on a signed arrangement $A=\left(V,\thinspace E,\thinspace l\right)$
is a game played by two cooperative players (call them Alice and Bob)
and a referee. Alice and Bob may agree on a prior strategy but cannot
communicate once the game starts. They both know the signed arrangement
$A$ that the game takes place on. 
\begin{enumerate}
\item The referee picks a random vertex $v$ in $V$ and one of the two
hyperedges containing it at random.
\item The referee sends $v$ to Alice and $e$ to Bob. 
\item Alice colors $v$ with one of two {}``colors'', $+1$ and $-1$,
and sends the color $f\left(v\right)$ to the referee. 
\item Bob colors each vertex of $e$ with one of two {}``colors'', $+1$
and $-1$, and send this coloring $c:\, e\rightarrow\left\{ +1,-1\right\} $
to the referee. 
\item The referee confirms that Alice or Bob have given valid colorings,
and that each label is either $+1$ or $-1$.
\item The referee checks that the parity of Bob's coloring matches the sign
of the edge in the arrangement , that $\prod_{u\in e}c\left(u\right)=l\left(u\right)$.
\item The referee checks that Bob's coloring is consistent with Alice's
coloring of $v$, meaning that $c\left(v\right)=f\left(v\right)$
\item If both the parity and consistency checks succeed, then Alice and
Bob have won the game, otherwise they have lost.
\end{enumerate}
\end{defn}
Note that this protocol differs from the one in the magic square where
both Alice and Bob colored a hyperedge, with Alice coloring rows and
Bob coloring columns. We use this modification, which is also used
for the magic pentagram game \cite{Mermin pentagram}, because it
generalizes to arrangements that do not share the magic square game's
property that its lines can be divided into two sets (rows and columns)
that do not intersect within each set.
\begin{prop}
When limited to classical strategies, Alice and Bob win with certainty
in the parity telepathy game on an arrangement only if it is classically
realizable. \end{prop}
\begin{proof}
First, we consider only deterministic strategies. Let $f\left(v\right)$
be the color Alice assigns to vertex $v$. In order to always pass
the consistency check, Bob must color each vertex as per $f\left(v\right)$.
Then, Bob passing every parity check is equivalent to $f\left(v\right)$
being a classical realization of the arrangement.

Since Alice and Bob cannot communicate after the protocol starts,
we may assume that any randomized strategy has Alice and Bob perform
all coin flips before the game. After the flips, the randomized winning
strategy would become a deterministic winning strategy.
\end{proof}
Quantum realizations give rise to quantum winning strategies, provided
an awkward technical caveat.
\begin{thm}
\label{thm:Magic implies winnable}On any magic signed arrangement
that has a quantum realization in which all the operators have all
real eigenvectors, if Alice and Bob may share quantum entanglement
in advance and perform quantum operations, then they have a strategy
that wins with certainty.\end{thm}
\begin{proof}
We begin by stating the winning strategy. Let $c$ be the quantum
realization of the arrangement on the finite-dimensional Hilbert-space
$\mathcal{H}$ and let $n=\dim\mathcal{H}$. Let $\left|1\right\rangle ,\dots,\left|n\right\rangle $
be a basis for $H$. Alice and Bob share between them the maximally
entangled state:\[
\Psi_{AB}=\sum_{i=1}^{n}\left|i\right\rangle \left|i\right\rangle \]
Then, to obtain the coloring of a vertex $v$, Alice or Bob performs
the indicated measurement $c\left(v\right)$ on their half of $\Psi_{AB}$
to obtain $+1$ or $-1$. Note that the order of Bob's measurements
doesn't matter, since all the measurements he must make commute.

We first check that this strategy always passes the parity check.
Within any hyperedge $e$, the assigned measurements commute, and
therefore are mutually diagonalize. In the diagonal basis, these measurements
are basis measurements with each basis element resulting in a $+1$
or $-1$ as labelled. Since the product of the measurements is $l\left(e\right)I$,
then the product of values corresponding to each basis element is
$l\left(e\right)$. So, the measured values satisfy the indicated
parity constraint. 

Next, we check consistency. We use the well-known fact that if two
parties each rotate their halves of a Bell state by an arbitrary real
orthogonal matrix, the Bell state remains fixed. When Alice and Bob
perform equal measurements that have real sets of eigenvectors, it
is equivalent to both performing an orthogonal followed by a standard
basis measurement that determines the outcome. The rotations leaves
the Bell pair invariant, after which the standard basis measurements
produce equal outcomes for Alice and Bob.
\end{proof}
Note that the Pauli operators in the quantum realization of the magic
square and pentagram that we provided in Figure \ref{fig:Magic square and pentagram with operators}
satisfy the real eigenvectors property and therefore suffice to win
those magic games. This will be the case for all quantum realizations
that we give. 

Note that the result of Theorem \ref{thm:Magic implies winnable}
does not imply that a nonmagic game has no winning quantum strategy
-- even if no quantum realization exists of the corresponding arrangement,
this does not rule out a quantum strategy of a completely different
type. However, a recent result of Cleve \cite{BCS} fills that gap
by proving the reverse direction in addition to the forwards one by
showing that any quantum strategy can be converted to a quantum realization.
We state a limited version of the result as it applies to generalized
Mermin games, translated into our terminology. 
\begin{thm}
[Cleve] If a parity telepathy game on an arrangement admits a perfect
quantum strategy using at most countable entanglement but no classical
strategy, then that arrangement is magic.
\end{thm}
As a result, our definition of magic corresponds to quantum pseudo-telepathy
being possibility for a Mermin game on that arrangement. By characterizing
which Mermin games are magic, we will then characterize which ones
exhibit quantum-pseudotelepathy.

\subsection{Intersection graphs}

We have define realized arrangements to have label vertices with measurements,
with hyperedges encoding constraints on these measurements. For our
main result, it will be convenient to switch to a dual representation
in which measurements label edges and vertices encode constraints.
In our drawings of arrangements, this corresponds to interchanging
the roles of points and lines. 
\begin{defn}
The \textbf{intersection graph }to an arrangement $A=\left(V,E\right)$
is the undirected graph $\left(V',E'\right)$ where $V'=E$, and there's
an edge between $e_{1},e_{2}\in V'$ for each vertex in the intersecton
$e_{1}\cap e_{2}$.
\end{defn}
Signed intersection graph also include a sign $\pm1$ on each vertex.
Applying a quantum realization of an arrangement to its intersection
graph, we obtain a labelling of the edges of the intersection graph
such that the labels on edges sharing a vertex commute and multiply
to the identity times the vertex's sign.

An intersection graph is equivalent to the hypergraph dual of $A$,
obtained by interchanging the roles of vertices and hyperedges. In
other words, we think of each vertex as the {}``hyperedge'' that
is the set of all the hyperedges that contain it. Since in an arrangement
each vertex lies on exactly two hyperedges, this dual hypergraph contains
only two-element hyperedges and is simply an undirected graph (technically
a multigraph).
\begin{example}
The intersection graphs of the magic square is the complete bipartite
graph on six vertices $K_{3,3}$, and the intersection graph of the
magic pentagram is the complete graph on five vertices $K_{5}$ (see
Figure \ref{fig:Magic square and pentagram dual}).
\end{example}
\begin{center}
\begin{figure}
\includegraphics[scale=0.5]{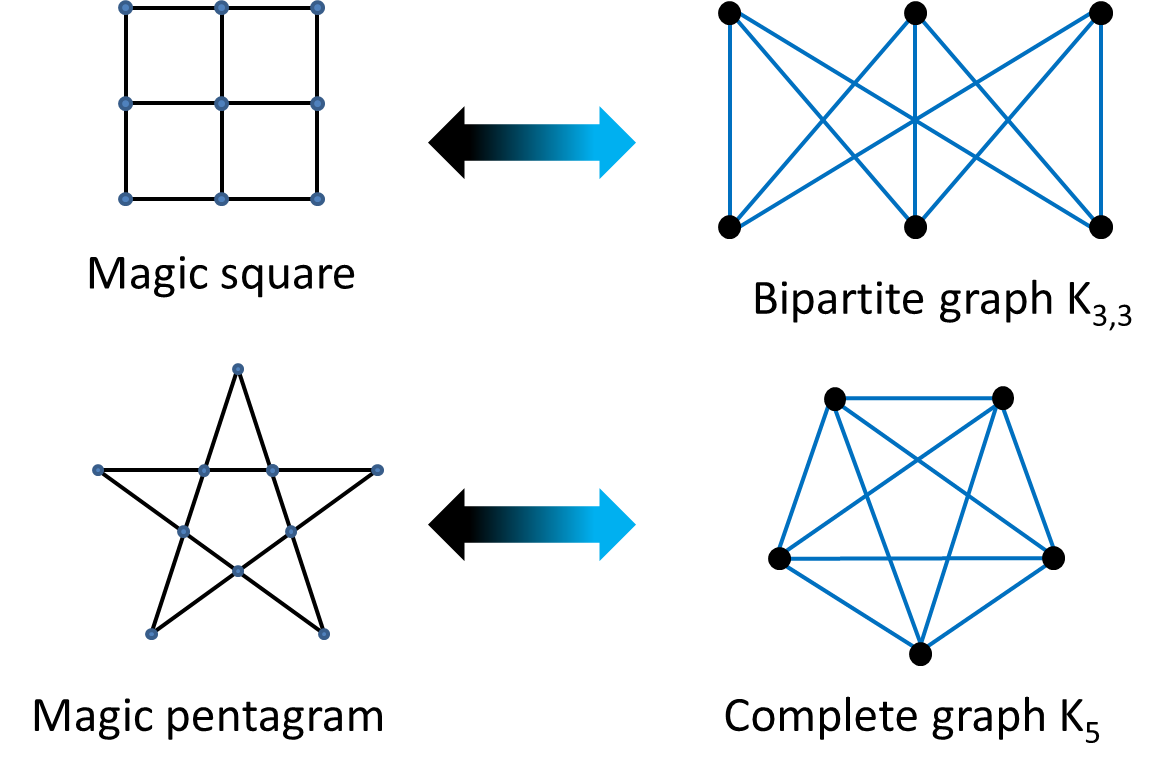}

\caption{The magic square has intersction graph $K_{3,3}$ and the magic pentagram
has intersection graph $K_{5}$.\label{fig:Magic square and pentagram dual}}

\end{figure}

\par\end{center}

\section{Main Result}

We will prove our main result. Recall that a graph is planar if it
can be embedded into the Euclidian plane so that no two edges intersect.
\begin{thm}
An arrangement is magic if and only if its intersection graph is not
planar. \label{thm:Main result}
\end{thm}
For ease of terminology, we'll call an intersection graph \textit{magic}
if its associated arrangement is magic. We will also talk about quantum
realizations of intersection graphs; these can be produced by taking
the operator vertex labels of an arrangement and transferring them
to the corresponding edges of the intersection graph.

We'll prove the two directions of Theorem \ref{thm:Main result} separately.

\subsection{Planar implies not magic}
\begin{thm}
\label{thm:planar implies nonmagic}If the intersection graph of an
arrangement is planar, then that arrangement is not magic.\end{thm}
\begin{proof}
We show that any realization of a signed version of this arrangement
must have even parity, and therefore is not magic. Out strategy will
to be to collapse the algebraic constraints on the vertex quantum
operators by repeatedly cancelling variable terms until we reach a
contradiction.

Recall that a quantum realization on a signed intersection graph labels
each edge with a measurement operator such that edges sharing a vertex
commute and multiply to the identity times a sign equal to that vertex's
label. Consider a planar embedding of this quantum realization of
this signed intersection graph. To contract an edge, delete that edge
and merge the two endpoint vertices into one, with edges that pointed
to one of these two vertices now pointing to the new vertex, and set
the new vertex's sign to be the product of the signs of the two merged
vertices. Merging may cause self-loops and multiple parallel edges.

We observe that contraction preserves the following properties:
\begin{itemize}
\item The product of the labels of edges around any vertex, in cyclic order
(say, without loss of generality, counterclockwise), equals identity
times the vertex's sign (a self-loop on a vertex will have its edge
label appear twice in the product.) We prove this here. This is only
nontrivial to check for the newly formed vertex. First, note that
if label the operators $M_{1},\dots,M_{n}$ going in a circle from
any starting point, if $M_{1}M_{2}\dots M_{n}=I$, then for any starting
point, $M_{k}M_{k+1}\dots M_{n}M_{1}M_{2}M_{k-1}=I$, since \begin{eqnarray*}
M_{k}M_{k+1}\dots M_{n}M_{1}M_{2}M_{k-1} & = & \left(M_{k-1}\dots M_{2}M_{1}\right)\left(M_{1}M_{2}\dots M_{n}\right)\left(M_{1}\dots M_{k-1}\right)\\
 & = & \left(M_{k-1}\dots M_{2}M_{1}\right)I\left(M_{1}\dots M_{k-1}\right)\\
 & = & I\end{eqnarray*}
A similar result follows with $-I$ in place of $I$. Now, consider
an edge labelled by an operator $X$, and let the labels around its
endpoint vertices be $M_{1},M_{2},\dots M_{m},X$ and $X,N_{1},N_{2},\dots N_{n}$
respectively going counterclockwise, and let the signs of the two
vertices be $\alpha_{M}$ and $\alpha_{N}$. Then, \begin{eqnarray*}
M_{1}\dots M_{n}X & = & \alpha_{M}I\\
XN_{1}\dots N_{n} & = & \alpha_{N}I\end{eqnarray*}
Mutiplying these gives \[
M_{1}\dots M_{n}N_{1}\dots N_{n}=\alpha_{M}\alpha_{N}I\]
The left hand side is the cyclically ordered product of edge labels
around the newly formed vertex, and the right hand side is the identity
with the sign of the new vertex, so the invariant remains. Note that
it may no longer be true that operators whose edges share a vertex
commute.
\item The graph embedding remains planar.
\item The sign parity (product of all the vertex labels) remains the same. 
\end{itemize}
Since each contraction reduces the number of vertices by $1$, contracting
any sequence of edges eventually produces a graph with a single vertex.
The sign of this vertex equals the product of the labels of all vertices
of the original intersection graph, and therefore equals its parity.
We will show that this parity is $+1$. 

It is easy to check that removing any self-loop that does not enclose
anything in the planar embedding also preserves the stated invariants;
such a self-loop contributes its operator twice in sequence to a cyclic
product, which cancels. Since there's always an innermost self-loop,
we may repeatedly remove such self-loops until none remain. So, the
sign of this vertex must equal the empty product, or $+I$. But, since
the sign parity has been preserved throguhout the process, this implies
that the original arrangement has sign parity $+1$ and therefore
has a classical realization and is not magic.
\end{proof}
The preceding proof gives a constructive way to derive a contradiction
from the algebraic constraints of any nonmagic game by combining the
constraints by multiplication and cancelling pairs of adjacent variables.
This is the same as the substitution approach from \cite{BCS}, first
used in a proof by Speelman\cite{Triangle}. See the section Discussion
for more on the connection.

\subsection{Nonplanar implies magic}

In this section, we prove the forward direction of the main result
(Theorem \ref{thm:Main result}).
\begin{thm}
If the intersection graph of an arrangement is nonplanar, then the
arrangement is magic.
\end{thm}
This proof will come in two pieces. First, we'll show that if an intersection
graph contains a magic intersection graph as a topological minor,
then it is magic. (Recall that we're saying an intersection graph
is magic as a shorthand for its associated arrangement being magic.)
We then using the well-known theorem of Pontyagin and Kuratowski that
any nonplanar graph contains either the complete graph $K_{5}$ or
the bipartite complete graph $K_{3,3}$ as a topological minor, and
that both of these intersection graphs are magic.
\begin{defn}
A graph $H$ is a \textbf{topological minor} of $G$ if $G$ has a
subgraph that is isomorphic to a subdivision of $H$, where a subdivision
is obtained by replacing each edge by a simple path of one or more
edges.
\end{defn}
Note that {}``subgraph'' as used in the above definition alows both
deleting edges and deleting vertices. By considering the isomorphism
explicitly, we obtain the following equivalent definition.
\begin{defn}
A graph $H$ is a topological minor of $G$ if there is an embedding
of $H$ in $G$ that consists of an injective map $\phi$ that takes
each vertex $v$ of $H$ to a vertex $\phi\left(v\right)$ of $G$,
and a map from each edge $\left(u,v\right)$ of $H$ to a simple path
from $\phi\left(u\right)$ to $\phi\left(v\right)$ in $G$, such
that these paths are disjoint except on their endpoints.\end{defn}
\begin{thm}
\label{thm:Topological minor transfers magic}If an intersection graph
$H$ is a topological minor of an intersection graph $G$, then $H$
being magic implies that $G$ is magic.\end{thm}
\begin{proof}
We will give a construction to turn a quantum realization of $H$
into one of $G$. Choose arbitrarily some odd-parity signing of $H$,
and let $c$ be a quantum realization of the corresponding arrangement
over some Hilbert space. We will use the topological minor inclusion
map to assign a corresponding signing and quantum realization on $G$
on the same Hilbert space, as follows:
\begin{enumerate}
\item For each vertex of $H$, label the corresponding vertex of $G$ with
the same sign. Label all other vertices of $G$ as $+1$.
\item For each edge of $H$, label each edge of the corresponding path in
$G$ with the same quantum operator. Label the remaining edges of
$G$ as $I$.
\end{enumerate}
We now show that this gives a quantum realization of $G$. We check
each of the required properties of a quantum realization, as interpreted
in the language of intersection graphs.
\begin{itemize}
\item Each measurement $M$ assigned to $G$ is either one in $H$ or the
identity, and therefore is Hermitian and has order $2$.
\item Each vertex of $G$ corresponding to a vertex of $H$ touches the
same measurement operators on its edges plus copies of the identity.
Therefore, these operators commute and have the same product as for
the vertex in $H$, which is labelled with the same sign.
\item Each vertex of $G$ that lies on a path that is the image on a edge
in $H$ touches two edges labelled with the same operator from that
edge in $H$, and possibly copies of the identity. These clearly commute
and multiply to $+I$, this vertex's label.
\item Each other vertex of $G$ only touches edges labelled as $I$, which
commute and have the correct product $+I$.
\end{itemize}
\end{proof}
\begin{thm}
(Pontryagin-Kuratowski) \label{thm:PK}A graph is nonplanar if and
only if it contains $K_{5}$ or $K_{3,3}$ as a topological minor.
\end{thm}
Other statements of this theorem use the stronger notion of graph
minors rather than topological minors, but in this case the notions
are equivalent. Note that we will only be using one direction of this
result, that every nonplanar graph contains $K_{5}$ or $K_{3,3}$. 
\begin{prop}
\label{pro:Excluded minors magic}The intersection graphs $K_{5}$
and $K_{3,3}$ are both magic.\end{prop}
\begin{proof}
The intersection graph $K_{5}$ corresponds to the magic pentagram
arrangement and $K_{3,3}$ to the magic square arrangement, which
we showed to be magic in Theorem \ref{thm:Square and pentagram are magic}.
\end{proof}
Combining Theorems \ref{thm:Topological minor transfers magic} and
\ref{thm:PK} and Proposition \ref{pro:Excluded minors magic} gives
the result stated at the start of the section.
\begin{thm}
If the intersection graph of an arrangement is nonplanar, then the
arrangement is magic.
\end{thm}

\section{Discussion }

The construction of quantum realizations for magic games suggests
that the magic square and magic pentagram are {}``universal'' for
magic games -- their quantum realizations give quantum realizations
for any magic arrangement. The fact the magic square uses only two-qubit
Pauli measurements and the magic pentagram uses three-qubit Paulis
(Theorem \ref{thm:Square and pentagram are magic}) therefore gives
a bound on the measurements necesarry to win a magic game. 
\begin{cor}
Any magic arrangement has a quantum realization with operators takes
from the three-qubit Pauli group.
\end{cor}

\begin{cor}
Any magic game can be won with certainty by players that share only
three Bell pairs of entanglement, and only use measurements from the
three-qubit Pauli group.
\end{cor}
Our proof of Theorem \ref{thm:planar implies nonmagic} that nonplanarity
implies magic demonstrates that the substitution approach suggested
by Speelman\cite{Triangle} and used by Cleve\cite{BCS} suffices
to prove that any constraint system has no quantum solution. In our
use, it gives a short and simple proof that a given Mermin-game is
nonmagic. The substitution approach proceeds by converting the constraint
system to equations specifying the products of measurement operators,
then repeatedly solving for an operator variable that appears once
in a constraint and substituting the solution whereever it appears,
and cancelling squares of operators until no variables the equation
$I=-I$ is reached.
\begin{cor}
In the language of \cite{BCS}, the substitution approach suffices
to prove the lack of a quantum satisfying assignment of any parity
binary constraint system in which every variable appears exactly twice,
if no quantum satisfying assignment exists.
\end{cor}
On parity binary constraint systems that \textit{do} have a quantum
solution, like that of the magic square game, the substitution approach
of course cannot prove the lack of a solution, and the barriers that
are encountered can be understood by the lack of a planar embedding
of the intersection graph making it impossible to get operator variables
adjacent so they can be cancelled.

\section{Future Work}

We chose a generalization of magic games in which every point lie
on exactly two hyperedges. It was this property that allowed us to
define intersection graphs (otherwise, they would be hypergraphs too)
and to use the excluded-minor characterization of graph planarity.
However, it is easy to imagine a generalization in which each point
may lie on any number of lines. For instance, Cleve's definition of
parity binary constraint systems\cite{BCS} allows each variable to
appear any number of times, corresponding to a point on any number
of lines. We would like to extend our results to such games. 

We have given one way to construct quantum realizations for magic
arrangements based on the magic square and pentagram as excluded minors.
Misusing complexity-theory terms, we've given {}``reductions'' of
non-magicness to the magic square and pentagram, showing these two
construction to be {}``complete'' for the class of magic arrangements.
This construction is not fully satisfying in that the quantum realization
for the magic square and pentagram appear ad-hoc; there seems to be
no a-priori reason they should exist. A more satisfying proof would
construct an assignment of operators directly from the nonplanar intersection
graph without looking for minors. We would like to understand the
set of all possible quantum realizations, along the vein of \cite{Square geometry,Pentagram geometry}
for the magic square and pentagram. We conjecture that any quantum
realization is isomorphic to one that uses Pauli operators.

\section{Acknowledgements}

I am grateful to Richard Cleve and Oded Regev for introducing me to
this problem and to the wonderful world of Mermin games. I would like
to thank many people for helpful discussions, including Scott Aaronson,
Adam Bouland, Richard Cleve, Andy Drucker, Matt Coudron, Oded Regev,
Cedric Lin, and Thomas Vidick. Thanks to Isaac Chuang and Piotr Indyk
for helpful suggestions. This work was partially supported by an NSF
fellowship.

\end{document}